\documentclass[11pt]{article}

\usepackage{mathptmx}
\usepackage{caption}

\usepackage{amsmath,amssymb,amsfonts,amsthm} 
\usepackage{mathrsfs,latexsym} 

\usepackage{calrsfs}
\DeclareMathAlphabet{\pazocal}{OMS}{zplm}{m}{n}

\usepackage{color}
\usepackage{cite}
\usepackage{epsfig}
\usepackage{graphicx}

\usepackage{bm}

\usepackage{cite}\medskip

\newtheorem{theorem}{Theorem}[section]

\numberwithin{equation}{section}

\newcommand{\II}{{\mathbb I}}
\newcommand{\CC}{{\mathbb C}}
\newcommand{\RR}{{\mathbb R}}

\newcommand{\NN}{{\mathbb N}}
\newcommand{\TT}{{\mathbb T}}

\newcommand{\Ac}{{\mathcal{A}}}
\newcommand{\Cc}{{\mathcal{C}}}
\newcommand{\Dc}{{\mathcal{D}}}

\newcommand{\Fc}{{\mathcal{F}}}
\newcommand{\Gc}{{\mathcal{G}}}

\newcommand{\bfi}{{\bm f}}
\newcommand{\bgi}{{\bm g}}

\newcommand{\bxi}{{\bm x}}
\newcommand{\byi}{{\bm y}}

\newcommand{\bPi}{{\bm P}}
\newcommand{\bQi}{{\bm Q}}

\newcommand{\Alg}{\pazocal{A}}                   
\newcommand{\BHil}[1]{\pazocal{B}(\pazocal{H}_#1)}   
\newcommand{\Dom}{\pazocal{D}}                   
\newcommand{\Group}{\pazocal{G}}
\newcommand{\Hil}{\pazocal{H}}                   
\newcommand{\Lag}{\pazocal{L}}              
\newcommand{\Weyl}{\pazocal{W}}

\newcommand{\Ad}[1]{\text{Ad} \, #1}

\def\eg{{\it e.g.\ }}
\def\ie{{\it i.e.\ }}

\def\etc{{\it etc}}

\begin{document} 

%

\title{ Classical dynamics, arrow of time, and genesis \\ of the
Heisenberg commutation relations  \\[5mm]
\Large In memory of Richard V.\ Kadison}
\author{\large Detlev Buchholz${}^{(1)}$ \ and 
\ Klaus Fredenhagen${}^{(2)}$ \\[5mm]
\small 
${}^{(1)}$ Mathematisches Institut, Universit\"at G\"ottingen, \\
\small Bunsenstr.\ 3-5, 37073 G\"ottingen, Germany\\[5pt]
\small
${}^{(2)}$ 
II. Institut für Theoretische Physik, Universit\"at Hamburg \\
\small Luruper Chaussee 149, 22761 Hamburg, Germany \\
}
\date{}

\maketitle

{\small 
\noindent {\bf Abstract.}
Based on the assumption that time evolves only in one direction
and mechanical systems can be described by Lagrangeans, a 
dynamical C*-algebra is presented for non-relativistic particles at
atomic scales. Without presupposing any quantization
scheme, this algebra is inherently non-commutative and
comprises a large set of dynamics. In contrast to other
approaches, the generating elements of the algebra are not
interpreted as observables, but as operations on the
underlying system; they describe the impact of 
temporary perturbations caused by the surroundings. In accordance with 
the doctrine of Nils Bohr, the operations carry individual names of
classical significance. Without stipulating from the outset their
``quantization'', their concrete implementation in the quantum world
emerges from the inherent structure of the algebra.
In particular, the Heisenberg commutation relations for position and
velocity measurements are derived from it. 
Interacting systems can be described
within the algebraic setting by a rigorous version of the interaction picture.  
It is shown that Hilbert space representations of the 
algebra lead to the conventional formalism of quantum mechanics, 
where operations on states are described by time-ordered 
exponentials of interaction potentials. It is also  
discussed how the familiar statistical 
interpretation of quantum mechanics can be recovered from 
operations.} \\[2mm]
{Mathematics Subject Classification: \ 81P16; 81Q99; 81S99 }  
 \\[2mm]
{Keywords: \ operations; \ arrow of time; \ dynamical C*-algebra } \ 

\section{Introduction}
\setcounter{equation}{0}

The appearance of quantum mechanics in the first half  
of last century was an important stimulus for the 
development of the theory of operator algebras. 
Initiated by John von Neumann, who invented the  
concept of ``rings of operators'' \cite{vNe} 
(now called von Neumann algebras), 
it was later generalized by Irving Segal, who advocated  
the usage of ``normed rings'' \cite{Se} (now called C*-algebras). 
Thenceforce, the theory of operator algebras 
has been a most lively interface 
between theoretical physics and pure mathematics.

\medskip 
In the second half of last century, Richard Kadison adopted a 
key role in the fruitful exchange between operator algebraists 
and quantum physicists. He conceived and organized two major conferences 
which greatly furthered the subject: Baton Rouge, Louisiana, 
in 1967, and Kingston, Ontario, in 1980. He also cooperated successfully
with physicists on conceptual and constructive problems in the
quantum setting \cite{DoKaKa,HaKaKa,HuKa}. And he consistently 
addressed mathematical problems of physical significance, such as 
the uniqueness of quantum states, fixed by 
maximal sets of commuting observables~\cite{KaSi}, 
the type of local observable algebras \cite{Ka2}, and 
the energy momentum spectrum of quantum fields \cite{Ka3}. 
As a matter of fact, two articles, raising a question
in the representation theory of the Heisenberg commutation relations,  
happened to be his very last scientific contributions \cite{KaLi,KaLiTh}. 
These fundamental relations are also subject of the present article.
  
\medskip 
The non-commutativity of algebras, appearing in quantum physics,  
has been traced back to the incommensurability of complementary 
observables. Relationships between observables, such 
as the Heisenberg relations, were brought to light 
by ingenious considerations on the basis of experimental facts and 
theoretical inspirations. Remarkably, it is possible to establish  
these relations also by quite elementary (classical) considerations, 
thereby complying with the doctrine of Niels Bohr that our interventions 
into the quantum world ought to be described in terms of ``common language''. 
These simplifications were recently uncovered in algebraic
quantum field theory in a search for a dynamical principle,
valid in that framework \cite{BuFr}. It is the aim of the present 
article to apply this novel scheme to quantum mechanics.

\medskip 
Instead of focussing on observables, generating the algebras,
we regard as primary entities the family of operations, describing the 
impact of temporary perturbations of the dynamics on the underlying 
states. Such perturbations typically arise in measurement arrangements, 
where forces can be manipulated. At this point the 
direction of time enters already at the microscopic level; because 
one can firmly state that some operation has happened earlier, 
respectively later, than another one. This fixes a natural ordering of
subsequent operations which can be cast into a causality
relation. It is this feature, the arrow of time, which we regard   
as the origin of non-commutativity. 

\medskip 
The dynamics enters in our
approach through the specification of a classical Lagrangean.
As is familiar from classical mechanics, it allows one to define variations of 
the corresponding action, called relative actions. They determine
corresponding variations of the operations, which can be expressed 
in terms of a dynamical relation. It is possible to justify this 
relation within the conventional framework of quantum theory \cite{BuFr}. 
Instead, we will take it here as input and show that it entails, 
together with the causality relation, the known formalism 
of quantum mechanics. 

\medskip
Within the present framework, all operations are labelled by functionals 
on the classical configuration space. They describe the envisaged 
perturbations of the quantum system  in ``common language''. 
Yet there is no \textit{a priori} quantization rule underlying 
our construction; the actual form of the operations at the quantum 
level is encoded in the relations between operations. 
In particular, the Heisenberg commutation relations 
turn out to be a consequence of them. 

\medskip 
It is note-worthy that the present approach can be
applied to quite arbitrary Lagrangean systems.
In this respect it resembles the Feynman path 
integral formalism. Yet, instead of having to deal with the 
subtle definition of functional integrals living on 
configuration spaces, we can work directly in a C*-algebraic setting. 
It covers the full set of operations and resultant observables of the system.  
This fact made it possible to construct dynamical C*-algebras in case of 
quantum field theories in arbitrary spacetime dimensions \cite{BuFr}, 
where corresponding  Feynman path integrals may not exist. Perhaps even 
more interestingly, this approach sheds also new light on the 
foundations of quantum theory.

\medskip 
Our article is organized as follows. In the subsequent section we
introduce our notation and recall some basic facts from 
classical mechanics. Given any Lagrangean, we adopt in Sec.~2
the methods developed in \cite{BuFr} and construct for the case 
at hand a dynamical group of operations. By standard methods we 
proceed from this group to a corresponding dynamical C*-algebra
and discuss some of its general properties. In Sec.~3 we 
carry out the steps devised in \cite{BuFr}
and show by methods developed there
that the resulting C*-algebra contains in the 
non-interacting case unitary exponentials of the position 
and momentum operators (Weyl operators),
satisfying the Heisenberg relations. By relying on an 
abstract version of the interaction picture, 
we find that also in the present case 
the algebras obtained for different 
Lagrangeans can be embedded into each other by injective 
homomorphisms. Sec.~4 contains the proof 
that the dynamical algebras are irreducibly and regularly 
represented in the Schr\"odinger representation by time-ordered 
exponentials of functions of the position and momentum
operators, described in terms of 
the classical theory. In Sec.~5 we show how the
standard statistical interpretation of quantum mechanics can
be derived from the dynamical C*-algebra 
without having to rely from the outset on spectral projections. 
The article concludes with a brief summary and outlook.

\section{Classical mechanics}
In this section, we recall notions from classical mechanics 
and introduce our notation.
We consider a system of $N$ classical point particles in
$s$-dimensional configuration space $\RR^s$. Their positions
are subsumed by vectors
$\bxi \doteq (\bxi_1, \dots \bxi_N) \in \RR^{s N}$ and 
the scalar product in $\RR^{s N}$ is given by 
$\bxi \, \byi$. 
The possible motions (orbits) of the particles are described by arbitrary 
smooth functions $\bxi : \RR \rightarrow \RR^{sN}$,
depending on time. They form a space denoted 
by $\Cc$. We also consider a subspace $\Cc_0 \subset \Cc$
of functions $\bxi_0$ having compact support in time;
so they form closed loops about the origin of $\RR^{s N}$.
Time derivatives are indicated by a dot, 
$\dot{\bxi}$. In order to
simplify the notation, we assume that all particles have
the same mass, which is put equal to~$1$. 

\medskip
Of primary interest in our approach is a space $\Fc$  of
localized (in time) functionals $F : \Cc \rightarrow \RR$. 
These are functionals of the form
$\bxi \mapsto F[\bxi] \doteq \int \! dt \, F(\bxi(t))$, where 
\begin{equation} \label{e2.1}
  F(\bxi(t)) = \bfi_0(t) \, \bxi(t) \, +  \, \sum_k g_k(t) V_k(\bxi(t)) \, ;
\end{equation}
here $\bfi_0 \in \Cc_0$ is a fixed loop, 
$g_k \in \Dc(\RR)$ are test functions with compact support,
and $V_k : \RR^{sN} \rightarrow \RR$ are continuous, bounded
functions, describing perturbations of the system. 
These functionals can be shifted by loops 
$\bxi_0 \in \Cc_0$. The shifts  are given by
\[
F^{\bxi_0}[\bxi] \doteq F[\bxi + \bxi_0] \, , \quad \bxi_0 \in \Cc_0 \, .
\]
The functionals $F$ are in general non-linear, but they satisfy 
the additivity relation 
\begin{equation} \label{e2.2} 
F[\bxi_1 + \bxi_2 + \bxi_3] = F[\bxi_1 + \bxi_3]
-F[\bxi_3] + F[\bxi_2 + \bxi_3] 
\end{equation}
whenever $\bxi_1$ and $\bxi_2$ have disjoint supports. This becomes 
evident if one splits the time axis into three 
disjoint pieces 
consisting of the support of $\bxi_1$, the support of $\bxi_2$, and their
common complement. 

\medskip 
The support of the functionals on the time axis
is defined as the set of points for which there exist loops $\bxi_{0}$, 
having support in arbitrarily small neighbourhoods of these instances
of time, such that $F \neq F^{\bxi_0}$. Thus, assuming that the 
potentials $V_k$ 
are linearily independent and disregarding constant potentials, 
which lead to functionals with empty support, the
support of a functional~$F$ 
is equal to the union of the supports of the underlying
loops $\bfi_0$ and
test functions~$g_k$. We will say that a 
functional $F_1$ lies in the future of $F_2$ if the support of~$F_1$
happens to be later than that of $F_2$ with regard 
to the chosen time direction. 

\medskip
We consider Lagrangeans $\Lag$ of the form 
\begin{equation} \label{e2.3}
t \mapsto \Lag(\bxi(t)) = 1/2 \ \dot{\bxi}(t)^2 - V_I(\bxi(t)) \, , 
\quad \bxi \in \Cc \, ,
\end{equation}
where $V_I$ describes some continuous, bounded 
interaction potential between the 
particles. The time integral of a 
Lagrangean determines the action of the 
underlying mechanical systems. 
Given $\Lag$ and any loop $\bxi_0 \in \Cc_0$,  the 
variations of the corresponding action are defined by   
\[
\delta \Lag(\bxi_0)[\bxi] \doteq 
\int \! dt \, \chi(t) \big( \Lag(\bxi(t) + \bxi_0(t)) - \Lag(\bxi(t)) \big) \, ,
\]
where $\chi$ is any test function which is equal to $1$ on the support
of $ \bxi_0$. By a partial integration one finds that the resulting
functional has the form given in \eqref{e2.1},
\[
\delta \Lag(\bxi_0)[\bxi] = \int \! dt \, \chi(t)
\big(- \ddot{\bxi}_0(t) \, \bxi(t) + (1/2) \, \dot{\bxi}_0(t)^2
- V_I(\bxi(t) + \bxi_0(t)) + V_I(\bxi(t))  \big) \, .
\]
In particular, it does not depend on the choice of $\chi$ within the
above limitations. If the gradient $\partial V_I$ of the
potential exists, then the stationary points of the action 
with regard to arbitary variations $\bxi_0$ determine the Euler-Lagrange 
equation 
$$
\ddot{\bxi} + \partial V_I(\bxi) = 0 \, .
$$
Their solutions describe the actual motions 
(orbits) of the mechanical system,

\medskip 
We will have to consider the action of the propagators 
(Green's functions) of the 
differential operator $K \doteq  - \frac{d^2}{dt^2} $ 
in this equation on given loops; its sign is a 
matter of convenience. Adopting notation and
terminology used in analogy to \cite{BuFr},
the kernels of the retarded, respectively advanced, 
propagator are continuous functions of time given by
\[
t,t' \mapsto \Delta_R(t,t') = - \Theta(t-t')(t -t') \, ,
\qquad 
t,t' \mapsto \Delta_A(t,t') = \Theta(t'- t)(t - t') \, , 
\]
where $\Theta$ denotes the Heaviside step function. Their mean 
is denoted by 
\[
t,t' \mapsto \Delta_D(t,t') \doteq (1/2) 
\big( \Delta_R(t,t') +  \Delta_A(t,t') \big) = - 1/2 \, |t - t'| \, ,
\]
and their difference is the commutator function  given by
\[
t,t' \mapsto \Delta(t,t') \doteq  \Delta_R(t,t') - 
\Delta_A(t,t') = t' - t \, .
\]
These Green's functions, when acting on loops, satisfy the equations  
\[
K \, \Delta_R = \Delta_R \, K = 1 \, , \quad
K \, \Delta_A = \Delta_A \, K = 1 \, , \quad
K \, \Delta_D = \Delta_D \, K = 1 \, , \quad
K \, \Delta = \Delta \, K = 0 \, .
\]

\medskip
We conclude this section by noting that we regard the particles 
as distinguishable. In case they are indistinguishable,  one has
to restrict attention to functionals on $\Cc$ which are symmetric
with respect to permutations of the particle indices. 

\section{The dynamical C*-algebra}
\setcounter{equation}{0}

We turn now to the definition of the dynamical C*-algebra, describing
the mechanical system. As mentioned in the introduction, we adopt
the scheme which has been established in \cite{BuFr} in the context 
of quantum field theory. For the sake of a coherent exposition, we 
recall here this simple construction. It is the primary 
purpose of the present section to highlight the fact that the dynamical 
C*-algebra is entirely based on classical concepts without 
imposing from the outset any quantization conditions 
for observables. 

\medskip 
Given a Lagrangean $\Lag$, one constructs in a first step a dynamical 
group $\Group_\Lag$. It is the free group generated by elements $S(F)$, 
modulo certain specific dynamical and causal relations. These elements  
are labelled by elements $F$ of 
the space of functionals $\Fc$, introduced in the preceding section.

\medskip
\noindent \textbf{Definition:} Let $\Lag$ be a Lagrangean of the form
given in equation \eqref{e2.3}. The corresponding dynamical group $\Group_\Lag$
is the free group generated by symbols $S(F)$, $F \in \Fc$, modulo the 
relations \\[1mm]
(i) \ \ \
$S(F) = S(F^{\bxi_0} + \delta \Lag(\bxi_0))$ \ for all \ $\bxi_0 \in \Cc_0$, 
$F \in \Fc$ \\[1mm]
(ii) \ $S(F_1 + F_2 + F_3) = S(F_1 + F_3) \, S(F_3)^{-1} \, S(F_2 + F_3)$
\ for arbitrary functionals $F_3$, provided $F_1$ lies in the future 
of $F_2$,  

\medskip
The first equality encodes dynamical information. It describes how a 
variation of the action affects the functionals. If 
$F = 0$ one obtains $S(\delta \Lag(\bxi_0)) = S(0)$ for 
$\bxi_0 \in \Cc$, where without loss of generality we put $S(0) = 1$. 
These equations are, within the present setting, 
the analogue of the Euler-Lagrange equation in classical mechanics.

\medskip
The second equality describes the impact of the arrow of time
on the causal properties of the theory. This causality relation 
corresponds to equation~\eqref{e2.2} within the present setting, where  
the chosen order of the first (later) and last (earlier) term is a matter of
common convention.

\medskip
We shall show in the subsequent section that the group $\Gc_\Lag$ is 
inherently non-commutative, \ie it has non-commutative 
representations. An important parameter entering in this context 
is determined by the constant functionals
$F_\textrm{h} : \Cc \rightarrow \RR$ which, for $\textrm{h} \in \RR$, are 
given by $F_\textrm{h}[\bxi] \doteq \textrm{h}$, $\bxi \in \Cc$. Since constant
functionals have empty support, the  causality relation 
for $F_3 = 0$ implies
$S(F) S(F_\textrm{h}) = S(F + F_\textrm{h}) = S(F_\textrm{h}) S(F)$, \ie the 
elements $S(F_\textrm{h})$ lie in the center of $\Gc_\Lag$. As we shall
see, they set the scale of Planck's constant, which we 
put equal to $1$ (atomic units). 
 
\medskip
The passage from a group to a C*-algebra is a standard procedure,
which we briefly recall here for the case at hand. 
We proceed first from $\Gc_\Lag$ to the corresponding group algebra 
$\Ac_\Lag$ over $\CC$. It is by definition the
complex linear span of the elements $S \in \Gc_\Lag$. 
For notational convenience, we also fix the central elements 
corresponding to the constant functionals, putting
$S(F_\textrm{h}) \doteq e^{i\textrm{h}} 1$, $\textrm{h}
 \in \RR$. The adjoint of the 
elements of $\Ac_\Lag$ is defined by putting
$(\sum c \, S)^* \doteq \sum \overline{c} \, S^{-1}$ and the
multiplication in  $\Ac_\Lag$  is inherited from $\Gc_\Lag$ by the
distributive law.

\medskip
For the construction of a C*-norm on $\Ac_\Lag$, we proceed from the
fact that there exists a functional $\omega$ on this algebra 
which is obtained by linear extension from the defining equalities
$\omega(S) = 0$ for $S \in \Gc_\Lag \backslash \{\TT 1\}$ and $\omega(1) = 1$. 
So for any choice of a finite number of different 
elements $S_i \in \Gc_\Lag \backslash \{\TT 1\}$, $i = 1, \dots , n$,  
and $S_0 = 1$ one has
\[
\omega\Big( \big(\sum_{i} c_i S_i \big)^* \, 
\big(\sum_{j} c_j S_j \big) \Big)
= \sum_{i,j} \overline{c}_i \, c_j \, \omega(S_i^{-1} S_j)
= \sum_{i} |c_i|^2 \geq 0 \, .
\]
This shows that, apart from the zero element, the functional 
$\omega$ has positive values on positive
elements of $\Ac_\Lag$, \ \ie it is a faithful state.
Thus, putting 
\[
\| A \|^2 \doteq \sup_{\omega'} \, \omega'(A^* A) \, , \quad A \in \Ac_\Lag \, ,
\]
where the supremum extends over all states $\omega'$ on $\Ac_\Lag$,
one obtains a C*-norm on $\Ac_\Lag$. Note that 
the supremum exists since the elements of $\Ac_\Lag$ are finite
linear combinations of unitary operators. 
The completion of $\Ac_\Lag$ 
with regard to this norm is a C*-algebra, which will be denoted
by the same symbol. 

\medskip
\noindent \textbf{Definition:} Given a Lagrangean $\Lag$, the 
corresponding dynamical algebra $\Ac_\Lag$ is the C*-algebra
determined by the group $\Gc_\Lag$, as explained above.

\section{Heisenberg commutation relations and dynamics}
\setcounter{equation}{0}

As a first application of our framework, we discuss the case of 
non-interacting particles, which are described by the Lagrangean 
\[
t \mapsto \Lag_0(\bxi(t)) = (1/2) \, \dot{\bxi}(t)^2 \, , 
\quad \bxi \in \Cc \, .
\]
The corresponding dynamical C*-algebra is $\Ac_{\Lag_0}$
and its generating unitary operators are denoted by $S_{\Lag_0}$. 
For the proof that this algebra contains 
operators satisfying the Heisenberg relations, we 
consider for given loop functions $\bfi_0 \in \Cc_0$ the functionals
\begin{equation} \label{e4.1}
F_{\bfi_0}[\bxi] \doteq \langle \bfi_0, \bxi \rangle + 
(1/2) \, \langle \bfi_0, \Delta_D \, \bfi_0 \rangle \, ,
\quad \bxi \in \Cc \, .
\end{equation}
Here we made use of the notation 
$ \langle \bfi_0 , \bxi \rangle \doteq 
\int \! dt \, \bfi_0(t) \bxi(t) $;   
the propagator~$\Delta_D$ was defined at the end of Sec.~2. 
We then define corresponding unitary operators  in $\Ac_{\Lag_0}$,
putting
\begin{equation} \label{e4.2}
W(\bfi_0) \doteq S_{\Lag_0}(F_{\bfi_0}) \, , \quad \bfi_0 \in \Cc_0 \, .
\end{equation}

\medskip 
As we shall see, these operators satisfy the Heisenberg 
commutation relations
in Weyl form. For the proof we need to have a closer look at the
underlying functionals.
Picking any loop $\bxi_0 \in \Cc_0$, putting its 
second time derivative $K \bxi_0$ into the functional, and
recalling that $\Delta_D \, K =  1$, we obtain 
\[
F_{K \bxi_0}[\bxi] 
= - \langle \ddot{\bxi}_0, \bxi \rangle - 
(1/2) \langle \ddot{\bxi}_0, \bxi_0 \rangle
=  - \langle \ddot{\bxi}_0, \bxi \rangle + 
(1/2) \langle \dot{\bxi}_0, \dot{\bxi}_0 \rangle
= \delta \Lag_0(\bxi_0)[\bxi] \, . 
\]
Thus the dynamical relation in $\Ac_\Lag$
implies $W(K \bxi_0) = S_{\Lag_0}(F_{K \bxi_0}) = 1$.

\medskip
Next, given $\bfi_0 \in \Cc_0$, let 
$\bfi_0 = \bfi_0^{\, \prime} + K \bxi_0$ be any decompostion  
with $\bfi_0^{\, \prime}, \bxi_0 \in \Cc_0$.
Plugging it into the functional gives 
\begin{align*}
F_{\bfi_0}[\bxi] & = \langle (\bfi_0^{\, \prime} + K \bxi_0) \, , \, \bxi \rangle +
(1/2) \langle (\bfi_0^{\, \prime} + K \bxi_0), \, \Delta_D \, 
(\bfi_0^{\, \prime} + K \bxi_0) \rangle  \\
& = \big( \langle \bfi_0^{\, \prime}\, , \, \bxi + \bxi_0 \rangle 
+ (1/2) \langle \bfi_0^{\, \prime} \, , \, \Delta_D \ \bfi_0^{\, \prime} 
\rangle \big) +
\big(\langle K \bxi_0 \, , \, \bxi \rangle 
+ (1/2) \langle K \bxi_0 \, , \, \Delta_D \, K \bxi_0 \rangle \big) \\
& = F_{\bfi_0^{\, \prime}}^{\, \bxi_0}[x] +
 F_{K \bxi_0}[x] = F_{\bfi_0^{\, \prime}}^{\, \bxi_0}[x] + \delta \Lag(\bxi_0)[x] \, .
\end{align*}
The dynamical relation in $\Ac_\Lag$, together with the preceding
equality, imply
\[
W(\bfi_0) = S_{\Lag_0}(F_{\bfi_0}) = S_{\Lag_0}(F_{\bfi_0^{\, \prime}}^{\, \bxi_0} + 
\delta \Lag_0(\bxi_0))
= S_{\Lag_0}(F_{\bfi_0^{\, \prime}}) = W(\bfi_0^{\, \prime}) \, ,
\]
which generalizes the relation $W( K \bxi_0) = 1$, obtained 
in the preceding step. 

\medskip
Finally, we pick two arbitrary loop functions $\bfi_{0}, \bgi_{0} \in \Cc_0$
and choose a decomposition $\bfi_{0} = \bfi_0^{\, \prime} + K \bxi_{0}$
such that $\bfi_0^{\, \prime}$ lies in the future of $\bgi_{0}$. 
That such a decomposition exists can be seen as follows. Choose
a smooth step function $\chi$ which has support in the future of 
$\bgi_{0}$ and is equal to $1$ at large times. Putting 
$\bfi_0^{\, \prime} \doteq K \chi \Delta_R \bfi_{0}$ and
$\bxi_{0} \doteq (1 - \chi) \Delta_R \bfi_{0}$,
one obtains loop functions with the desired support 
properties and  
$\bfi_0^{\, \prime} + K \bxi_{0} = K \Delta_R \bfi_{0} = \bfi_{0}$, as
claimed. This equality, together with the preceding results 
and the causality relation imply 
\[
W(\bfi_{0}) W(\bgi_{0}) = 
S_{\Lag_0}(F_{\bfi_{0}}) S_{\Lag_0}(F_{\bgi_{0}}) =  S_{\Lag_0}(F_{\bfi_0^{\, \prime}}) 
S_{\Lag_0}(F_{\bgi_{0}})
= S_{\Lag_0}(F_{\bfi_0^{\, \prime}} +F_{\bgi_{0}}) \, .
\]
Now, for $\bxi \in \Cc$, 
\begin{align*}
& F_{\bfi_0^{\, \prime}}[\bxi] + F_{\bgi_{0}}[\bxi] = 
\langle \bfi_0^{\, \prime}, \bxi \rangle + (1/2) 
\langle \bfi_0^{\, \prime}, \Delta_D \, \bfi_0^{\, \prime} \rangle + 
\langle \bgi_{0}, \bxi \rangle + (1/2) 
\langle \bgi_{0}, \Delta_D \, \bgi_{0} \rangle \\
& = \langle (\bfi_0^{\, \prime} + \bgi_{0}), \bxi \rangle +
(1/2) \langle (\bfi_0^{\, \prime} + \bgi_{0}), \Delta_D \, 
(\bfi_0^{\, \prime} + \bgi_{0}) \rangle - 
\langle \bfi_0^{\, \prime}, \Delta_D \, \bgi_{0} \rangle \\
& = F_{\bfi_0^{\, \prime} + \bgi_{0}}[\bxi] - \langle \bfi_0^{\, \prime}, 
\Delta_D \, \bgi_{0} \rangle \, .
\end{align*}
For the last term, being a constant functional, 
we obtain in view of the support properties of $\bfi_0^{\, \prime}, \, \bgi_{0}$ 
and the fact that $K \Delta = 0$  
\begin{align*}
\langle \bfi_0^{\, \prime}, \Delta_D \, \bgi_{0} \rangle
& = (1/2) \langle \bfi_0^{\, \prime}, \Delta_R \, \bgi_{0} \rangle
= (1/2) \langle \bfi_0^{\, \prime}, \Delta \, \bgi_{0} \rangle \\
& = (1/2) \langle (\bfi_0^{\, \prime} + K \bxi_{0}), \Delta \, \bgi_{0} \rangle
= (1/2) \langle \bfi_{0}, \Delta \, \bgi_{0} \rangle \, .
\end{align*} 
Bearing in mind the results of the preceding step, this 
implies  
\[ 
S_{\Lag_0}(F_{\bfi_0^{\, \prime}} + F_{\bgi_{0}}) =
e^{-(i/2)\langle \bfi_{0}, \Delta \, \bgi_{0}  \rangle} \, 
S_{\Lag_0}(F_{\bfi_0^{\, \prime} + \bgi_{0}}) = 
e^{-(i/2)\langle \bfi_{0}, \Delta \, \bgi_{0}  \rangle} \, 
W(\bfi_{0} +\bgi_{0})  \, ,
\] 
so this operator coincides with the product $W(\bfi_{0}) W(\bgi_{0})$.
The results obtained so far are summarized in the subsequent 
theorem.  

\begin{theorem} \label{t4.1}
Let $W(\bfi_0) \in \Ac_{\Lag_0}$, $\bfi_0 \in \Cc_0$, be the
unitary operators, defined in equation~\eqref{e4.2}. Then 
\[ 
W(\bfi_{0}) W(\bgi_{0}) = 
e^{-(i/2)\langle \bfi_{0}, \Delta \, \bgi_{0}  \rangle} \, 
W(\bfi_{0} +\bgi_{0}) \, , \quad \bfi_{0}, \bgi_{0} \in \Cc_0 \, .
\]
Moreover, there hold the dynamical relations $W(K \bxi_{0}) =  1$,
$\bxi_{0} \in \Cc_0$. 
\end{theorem}
It follows from this theorem that 
the operators $W(\bfi_0)$, $\bfi_0 \in \Cc_0$,
form a unitary Lie group, the Weyl group $\Weyl$. 
Proceeding to its Lie algebra and 
denoting the corresponding generators by 
$\langle \bxi_0, \bQi \rangle$ and $1$, 
the dynamical relations imply   
$\langle K \bxi_0, \bQi \rangle = 0$, $\bxi_0 \in \Cc_0$, hence 
\[
\bQi(t) = \bQi + t \dot{\bQi} \, , \quad t \in \RR \, ,
\]
where we have absorbed possible constants into
$\bQi$ and $\dot{\bQi}$. The non-trivial commutators of the generators are \ 
$[\langle \bfi_{0}, \bQi \rangle, \langle \bgi_{0}, \bQi \rangle]
= i \langle \bfi_{0}, \, \Delta \, \bfi_{0} \rangle 1 \, .
$
They yield for the components of the generators the 
commutation relations
\[
[\bQi_k, \, \dot{\bQi}_l] = i \, \delta_{k l} 1 \, , \quad
[\bQi_k, \, \bQi_l] = [\dot{\bQi}_k, \, \dot{\bQi}_l] = 0 \, , 
\quad k,l = 1, \dots, sN \, .
\]
Identifying $\bQi$ with position and the
velocity $\dot{\bQi}$ with momentum $\bPi$, these are the Heisenberg 
commutation relations for the corresponding quantum observables.
The dynamical relation is the solution of the Heisenberg equation 
in the absence of interaction. 

\medskip
We will show in Sec.\ 5 that the linear functionals 
$\bxi \mapsto L_{\bfi_0}[\bxi] \doteq  
\langle \bfi_0, \bxi \rangle$, which appeared in the preceding step,
give rise to unitaries $S_{\Lag_0}(L_{\bfi_0})$ which are the time-ordered 
exponentials of 
the corresponding generators $\langle \bfi_0, \bQi \rangle$.
Similarly, the 
unitaries $S_{\Lag_0}(F)$ for functionals~$F$ of the form~\eqref{e2.1} 
are the respective 
time-ordered exponentials, where the classical orbits $t \mapsto \bxi(t)$ 
are replaced by $t \mapsto \bQi(t)$.  
So, instead of representing the observable $F$ at the
quantum level, the operators 
$S_{\Lag_0}(F)$ describe the perturbations of the non-interacting dynamics,  
caused by their temporary action, in accordance with the interaction 
picture in quantum mechanics. 

\medskip
This insight enters in our subsequent arguments, where we compare the 
algebras $\Ac_\Lag$ for different Lagrangeans $\Lag$. Given any 
Lagrangean $\Lag_0$ (which may differ from the non-interacting one), 
we will show that the algebras $\Ac_\Lag$ for the perturbed Lagrangeans 
$\Lag = \Lag_0 - V_I$ can be embedded by an injective 
homomorphism into the algebra 
$\Ac_{\Lag_0}$ associated with $\Lag_0$. To this end we 
choose any interval $\II \subset \RR$ and a corresponding
smooth characteristic function $\chi$
which has support in a slightly larger interval 
$\hat{\II} \supset \II$. We then consider
the temporary perturbation of the Lagrangean $\Lag_0$, 
\[
t \mapsto \Lag_\chi(\bxi(t)) \doteq \Lag_0(\bxi(t)) - \chi(t) V_I(\bxi(t)) \, , 
\quad \bxi \in \Cc \, .
\] 
The corresponding relative action 
for $\bxi_0 \in \Cc_0$ is given by 
\[
\delta \Lag_\chi(\bxi_0) = \delta \Lag_0(\bxi_0) - V_I^{\bxi_0}(\chi) +
V_I(\chi) \, ,
\]
where $\bxi \mapsto V_I(\chi)[\bxi] \doteq   
\int \! dt \, \chi(t) V_I(\bxi(t))$. Note that 
$\delta \Lag_\chi(\bxi_0)$ coincides with the full
relative action $\delta \Lag(\bxi_0)$
for loops $\bxi_0$ having support in $\II$.
We define now the unitary operators in $\Ac_{\Lag_0}$ 
\[
S_{\Lag_\chi}(F) \doteq S_{\Lag_0}(-V_I(\chi))^{-1} \, S_{\Lag_0}(F - V_I(\chi)) \, ,
\quad F \in \Fc \, . 
\]
It is apparent that the operators $S_{\Lag_\chi}(F)$, $F \in \Fc$,
generate the algebra $\Ac_{\Lag_0}$.
Making use of the dynamical relation in $\Ac_{\Lag_0}$, we obtain
for $\bxi_0 \in \Cc_0$ 
\begin{align*}
S_{\Lag_\chi}(F^{\bxi_0} + \delta \Lag_\chi(\bxi_0)) & =
S_{\Lag_0}(-V_I(\chi))^{-1} \,
S_{\Lag_0}(F^{\bxi_0}  - V_I^{\bxi_0}(\chi) + \delta \Lag_0(\bxi_0)) \\ 
& = S_{\Lag_0}(-V_I(\chi))^{-1} \,
S_{\Lag_0}(F  - V_I(\chi)) = S_{\Lag_\chi}(F) \, .
\end{align*}
Similarly, if $F_1, F_2, F_3 \in \Fc$ are functionals such that 
$F_1$ lies in the future of $F_2$, the causal relations in 
$\Ac_{\Lag_0}$ imply
\begin{align*}
& S_{\Lag_\chi}(F_1 + F_3) \ S_{\Lag_\chi}(F_3)^{-1}  S_{\Lag_\chi}(F_2 + F_3) \\
& = S_{\Lag_0}(-V_I(\chi))^{-1} \, S_{\Lag_0}(F_1 + F_3 - V_I(\chi)) \, 
  S_{\Lag_0}(F_3 - V_I(\chi))^{-1} \, S_{\Lag_0}(F_2 + F_3 - V_I(\chi)) \\
& = S_{\Lag_0}(-V_I(\chi))^{-1} \,  S_{\Lag_0}(F_1 + F_2 + F_3 - V_I(\chi)) =
S_{\Lag_\chi}(F_1 + F_2 + F_3) \, .
\end{align*}
So the unitary operators $S_{\Lag_\chi} : \Fc \rightarrow \Ac_{\Lag_0}$ satisfy the
defining relations of the dynamical group $\Gc_{\Lag_\chi}$
for the Lagrangean $\Lag_\chi$. Moreover, $\Ac_{\Lag_\chi} = \Ac_{\Lag_0}$. 

\medskip 
We want to control the limit $\II \nearrow \RR$.  
To this end we restrict the unitaries $S_{\Lag_\chi}$ 
to functionals in $\Fc(\II) \subset \Fc$, 
having support in~$\II$. These restrictions generate 
a subgroup $\Gc_{\Lag_\chi}(\II) \subset \Gc_{\Lag_\chi}$. 
Since $\delta \Lag_\chi(\bxi_0) = \delta \Lag(\bxi_0)$ for 
loops $\bxi_0$ having support in $\II$, this subgroup is isomorphic to the 
group $\Gc_\Lag(\II)$, which is obtained by restricting the unitaries
$S_\Lag$, assigned to the full Lagrangean, to $\Fc(\II)$.
The resulting isomorphism 
$\beta_{\, \II, \chi} : \Gc_{\Lag_\chi}(\II) \rightarrow \Gc_\Lag(\II)$ 
extends to congruent linear combinations
of the group elements, forming algebras 
$\Ac_{\Lag_\chi}(\II) \subset \Ac_{\Lag_0}$ and $\Ac_\Lag(\II) \subset \Ac_\Lag$,
respectively. Denoting by $\| \, \cdot \, \|_{\Lag_0}$ and
$\| \, \cdot \, \|_\Lag$ the C*-norms on $\Ac_{\Lag_0}$ and $\Ac_\Lag$, 
we define norms $\|\beta_{\, \II, \chi}( \, \cdot \, )\|_\Lag$
on $\Ac_{\Lag_\chi}(\II)$ and 
$\|\beta_{\, \II, \chi}^{-1}( \, \cdot \, )\|_{\Lag_0}$ on $\Ac_\Lag(\II)$.
Because of the maximality of the original C*-norms, one has
$\|\beta_{\, \II, \chi}( \, \cdot \, )\|_\Lag \leq \| \, \cdot \, \|_{\Lag_0}$
and $\|\beta_{\, \II, \chi}^{-1}( \, \cdot \, )\|_{\Lag_0} \leq \| \, 
\cdot \, \|_\Lag$. This implies  
$\|\beta_{\, \II, \chi}( \, \cdot \, )\|_\Lag = \| \, \cdot \, \|_{\Lag_0}$
and $\|\beta_{\, \II, \chi}^{-1}( \, \cdot \, )\|_{\Lag_0} = \| \, \cdot \, \|_\Lag$.
It follows that the isomorphism 
$\beta_{\, \II, \chi} : \Ac_{\Lag_\chi}(\II) \rightarrow \Ac_\Lag(\II)$ 
extends to an isomorphims between the norm closures of these 
subalgebras of $\Ac_{\Lag_0}$, respectively $\Ac_\Lag$, which we
denote by the same symbols. 

\medskip
In the next step we need to determine the dependence of
the  isomorphisms $\beta_{\, \II, \chi}$ on the choice of the
smooth characteristic function $\chi$ for given 
interval $\II$. Let $\chi_1$ and $\chi_2$ be two such 
functions which both have support in $\hat{\II} \supset \II$
and are equal to $1$ in some neighbourhood of $\II$.
One then has 
$\chi_2 - \chi_1 = \chi_+ + \chi_-$, where $\chi_+$ has 
support in the future of $\II$ and $\chi_-$ in its past. 
Picking any functional $F \in \Fc(\II)$,  
it follows from the causality relation for the 
operators $S_{\Lag_0}$ that
\begin{align*}
& S_{\Lag_0}(F - V_I(\chi_2)) \\
& = S_{\Lag_0}(-V_I(\chi_+) -V_I(\chi_1 + \chi_-)) \,  
S_{\Lag_0}(-V_I(\chi_1 + \chi_-))^{-1}
S_{\Lag_0}(F - V_I(\chi_1 + \chi_-)) \, .
\end{align*}
In a similar manner one obtains
\[
S_{\Lag_0}(F - V_I(\chi_1 + \chi_-)) = S_{\Lag_0}(F - V_I(\chi_1)) \, 
S_{\Lag_0}(-V_I(\chi_1)^{-1} \, S_{\Lag_0}(-V_I(\chi_1) - V_I(\chi_-)) \, .
\]
Plugging these equalities into the defining equation
for $S_{\Lag_{\chi_2}}$, we obtain
\[
S_{\Lag_{\chi_2}}(F) =  S_{\Lag_{\chi_1}}(-V_I(\chi_-))^{-1} \, S_{\Lag_{\chi_1}}(F) \, 
S_{\Lag_{\chi_1}}(-V_I(\chi_-)) \, , \quad F \in \Fc(\II) \, .
\]
This relation shows that the isomorphisms 
$\beta_{\, \II, \chi_1}^{-1}, \beta_{\, \II, \chi_2}^{-1}$, mapping 
$\Ac_{\Lag} (\II)$ onto the subalgebras $\Ac_{\Lag_{\chi_1}}(\II)$,
respectively $\Ac_{\Lag_{\chi_2}}(\II)$, 
of $\Ac_{\Lag_0}$, are related by an inner automorphism of 
$\Ac_{\Lag_0}$. 
Putting $U_{\chi_2, \chi_1} \doteq  S_{\Lag_{\chi_1}}(-V_I(\chi_-))^{-1}$, 
the preceding equality implies 
\[
\text{Ad} \, U_{\chi_2, \chi_1} \circ  \beta_{\, \II, \chi_1}^{-1}
= \beta_{\, \II, \chi_2}^{-1} \, .
\]
Since the supports of $\chi_1, \chi_2$ are contained in $\hat{\II}$,
it is also clear that $U_{\chi_2, \chi_1} \in \Ac_{\Lag_0}(\hat{\II})$.

\medskip 
We choose now an increasing sequence of intervals 
$\II_n \subset \hat{\II}_n \subset \II_{n+1}$, which exhaust 
$\RR$, and corresponding smooth 
characteristic functions $\chi_n$ which are equal to $1$ on 
$\II_n$ and have support in $\hat{\II}_n$,  $n \in \NN$.
It follows from the definition of the algebras $\Ac_\Lag(\II)$
and isomorphisms $\beta^{-1}_{\II, \chi}$ that
$\beta^{-1}_{\II_{n+1}, \chi_{n+1}} \upharpoonright \Ac_\Lag(\II_n)
= \beta^{-1}_{\II_n, \chi_{n+1}}$, $n \in \NN$.
On the basis of the preceding results, we define 
isomorphisms $\gamma_{\ \II_n}$, putting
\[
\gamma_{\ \II_n} \doteq \text{Ad} \, (U_{\chi_{n},\chi_{n-1}} 
\cdots U_{\chi_2,\chi_1})^{-1} \circ
\beta_{\, \II_n, \chi_n}^{-1} \, , \quad n \in \NN + 1 \, .
\]
It follows from the support properties of the functions $\chi_n$ that
$\gamma_{\ \II_n}(\Ac_\Lag(\II_n)) \subset \Ac_{\Lag_0}(\hat{\II}_n)$.
Moreover, 
\begin{align*}
\gamma_{\ \II_{n+1}} \upharpoonright \Ac_\Lag(\II_n)
& = \text{Ad} \, (U_{\chi_{n},\chi_{n-1}} \cdots U_{\chi_2,\chi_1})^{-1} \circ
\text{Ad} \, U_{\chi_{n+1},\chi_{n}}^{-1} \circ \beta_{\II_n, \chi_{n+1}}^{-1}
\upharpoonright \Ac_\Lag(\II_n) \\
& = \text{Ad} \, (U_{\chi_{n},\chi_{n-1}} \cdots U_{\chi_2,\chi_1})^{-1} \circ
 \beta_{\II_n, \chi_n}^{-1} \ = \ \gamma_{\ \II_n} \, , \quad n \in \NN \, .
\end{align*}
Thus, for any given interval $\II_m$, the restrictions 
$\gamma_{\ \II_n} \upharpoonright \Ac_\Lag(\II_m)$ stay constant for 
$n \geq m$ and their range is contained in $\Ac_{\Lag_0}(\hat{\II}_m)$. 
Since $ \Ac_\Lag$ is the C*-inductive limit of its subalgebras 
$\Ac_\Lag(\II_m)$, $m \in \NN$, it follows that the limit 
$\gamma \doteq \lim_n \gamma_{\ \II_n}$ exists pointwise in norm
on $ \Ac_\Lag$ and has range in $\Ac_{\Lag_0}$. More explicitly, one 
has for any interval $\II \subset \RR$ and sufficiently large
$n \in \NN$
\begin{equation} \label{e4.3} 
\gamma(S_\Lag(F)) = 
\text{Ad} \, (U_{\chi_{n},\chi_{n-1}} \cdots U_{\chi_2,\chi_1})^{-1} 
(S_{\Lag_0}(F)) \, , \quad
F \in \Fc(\II) \, . 
\end{equation}
Recalling that the choice of Lagrangeans $\Lag_0, \Lag$ was arbitrary, we 
have arrived at the following theorem, relating the dynamical algebras  
attached to different dynamics.
\begin{theorem} \label{t4.2}
Let $\Lag_0, \Lag$ be Lagrangeans of the form 
given in equation \eqref{e2.3}. The algebra $\Ac_\Lag$ can 
be embedded into $\Ac_{\Lag_0}$ by the injective 
homomorphism $\gamma$ given in  
equation \eqref{e4.3}. Moreover, for any interval $\II \subset \RR$
there is some interval $\hat{\II} \supset \II$ such that 
$\gamma(\Ac_\Lag(\II)) \subset \Ac_{\Lag_0}(\hat{\II})$. 
\end{theorem}

\section{Representations}
\setcounter{equation}{0}

We turn now to the construction of representations of the dynamical 
algebras. It suffices to focus on representations of the algebra 
$\Ac_{\Lag_0}$ for the non-interacting Lagrangean $\Lag_0$.
According to Theorem \eqref{t4.2},
all other algebras $\Ac_\Lag$ can also be represented
on the underlying Hilbert spaces. 
In more detail, denoting by $(\pi_0, \Hil_0)$
a representation $\pi_0 : \Ac_{\Lag_0} \rightarrow \ \BHil{0}$ 
on a Hilbert space $\Hil_0$, one obtains a representation
$(\pi, \Hil_0)$ of $\Ac_\Lag$ on $\Hil_0$, putting 
$\pi \doteq \pi_0 \circ \gamma$,
where $\gamma$ is the injective homomorphism given in equation \eqref{e4.3}.

\medskip 
According to Theorem \ref{t4.1}, the algebra $\Ac_{\Lag_0}$ contains
operators which can be interpreted as exponentials of the position and
momentum operators $\bQi, \bPi$, subject to the free time evolution
$t \mapsto \bQi(t) = \bQi + t \bPi$. We therefore proceed to the Schr\"odinger 
representation of the canonical commutation relations on the
Hilbert space $\Hil_S$,
where we keep the notation $\bQi, \bPi$ for the concrete multiplication
and differential operators. On $\Hil_S$ we consider for any given functional 
$F \in \Fc$, cf.\ equation \eqref{e2.1},  the operator function 
\[
  t \mapsto F(\bQi(t)) =
  e^{\, i t (1/2 \bPi^2} F(\bQi) e^{\, - i t (1/2 \bPi^2} \, ,
\]
\ie we replace in the functional the classical motions
$t \mapsto \bxi(t) = \bxi_0 + t \dot{\bxi}_0$ 
by their quantum counterpart. The adjoint action of the
unitaries $t \mapsto  e^{\, i t (1/2) \, \bPi^2}$, involving the
free Hamiltonian, induces these time translations on $\Hil_S$. 
Because of the linear
terms appearing in $F$, the resulting operators are in general unbounded,
but the operators are densely defined on $\Hil_S$. 
It is our goal to construct the time-ordered exponentials of 
the integrated operator functions, formally given by 
\[
T(F) \doteq T \, exp \Big( i \int_{-\infty}^\infty \! dt \, F(\bQi(t)) \Big) \, ,
\]  
where $T$ denotes time ordering. This will be
accomplished in several steps.

\medskip
We begin by  considering the cases where the functions 
$t \mapsto F(\bxi(t))$, $\bxi \in \Cc$, are
uniformly bounded. Then the above operator
function is bounded and continuous in the
strong operator topology on $\Hil_S$, $t \in \RR$. Its time-ordered
exponential is given by the Dyson expansion \cite{ReSi}
\[
  T(F) =
  1 + \sum_{k=1}^\infty i^k \int_{-\infty}^\infty \! dt_1 \int_{-\infty}^{t_1}
  \! dt_2 \dots \int_{-\infty}^{t_{k-1}} \! dt_k \ 
    F(\bQi(t_1)) \  F(\bQi(t_2)) \cdots F(\bQi(t_k)) \, .
\]  
The integrals are defined in the strong operator topology 
and the series converges absolutely in norm since the
operator functions are bounded and have compact support.

\medskip
Next, let $F_1, F_2 \in \Fc$ be bounded functionals such that 
the support of $F_1$ lies in the future of~$F_2$. Then 
$t, s \mapsto F_1(\bQi(t)) \, F_2(\bQi(s)) 
= F_2(\bQi(s)) \, F_1(\bQi(t))  = 0$ for $s \geq t$.
Let $t_0$ be a point in time such that $F_1$ lies in its 
future and $F_2$ in its past.
This implies after a moments reflection that 
\begin{align*}
& T(F_1 + F_2) =  1 +
\sum_{n=1}^\infty i^n
\sum_{k + l = n}^\infty \\ 
& \int_{t_0}^\infty \! \! \! dt_1 
\dots \int_{t_0}^{t_{k-1}} \! \! \! dt_k \, 
\int_{-\infty}^{t_0} \! \! \! ds_1 
\dots \int_{-\infty}^{t_{k-1}} \! \! \! ds_l \ 
F_1(\bQi(t_1))  \cdots F_1(\bQi(t_k)) \ 
F_2(\bQi(s_1)) \cdots F_2(\bQi(s_l)) \\[1mm] 
& = T(F_1) T(F_2) \, .
\end{align*}
Now, with $F_1, F_2$ as before, let $F_3$ be an arbitrary 
bounded functional. We decompose $F_3$  
sharply into $F_3 = F_{3 +} + F_{3 -}$ such 
that $F_{3 +}$ has its support in the future of $F_2$, $F_{3 -}$, 
and $F_{3 -}$ in the past of~$F_1$, $F_{3 +}$. Note that this 
sharp decomposition does not cause any problems since the
respective time ordered integrals are well defined. 
Then, according to the preceding result, 
\begin{align*}
T(F_1 + F_2 + F_3) & = T(F_1 + F_{3 +}) T(F_2 + F_{3 -}) \\ 
& =  T(F_1 + F_{3 +}) T(F_{3 -}) T(F_{3 -})^{-1} T(F_{3 +})^{-1}
T(F_{3 +})  T(F_2 + F_{3 -}) \\
& = T(F_1 + F_3) T(F_3)^{-1} T(F_2 + F_3) \, . 
\end{align*}
This is the causal factorization relation, 
anticipated in the abstract setting. We note that the 
constant functionals $F_{\textrm{h}}$ can be realized by 
choosing a function $t \mapsto \textrm{h}(t)$
which satisfies $\int \! dt \,  \textrm{h}(t) = \textrm{h}$
and has arbitrary support, \eg in the complement of any 
other given functional. Plugging the functional
$t \mapsto F_{\textrm{h}}(\bQi(t)) \doteq \textrm{h}(t)$
into the definition of the time-ordered operators,
one obtains 
$T(F) T(F_{\textrm{h}}) = T(F + F_{\textrm{h}}) = T(F_{\textrm{h}}) T(F) $.

\medskip
In order to extend the operators $T(F)$ to all 
functionals in $\Fc$, we consider now for given
loop $\bfi_0 \in \Cc_0$ the linear operator functions
$t \mapsto L_{\bfi_0}(\bQi(t)) \doteq \bfi_0(t) \bQi(t)$.
These operators are unbounded. But since they are 
linear combinations of the position and momentum 
operators, all of their products have as common dense 
domain $\Dom_S \subset \Hil_S$ the span of the eigenvectors of the 
Hamiltonian $(\bPi^2 + \bQi^2)$ of the harmonic
oscillator. Products of the time translated operators
act continuously in time 
on the eigenvectors. As a matter of fact, since the 
loop functions have compact support, the Dyson expansion
exists on each member of $\Dom_S$ and converges 
pointwise in the strong topology to the time-ordered unitary 
exponential $T(L_{\bfi_0})$.  
That operator can also be constructed by solving the differential equation 
\[
{\frac{d}{dt}} 
T(L_{\bfi_0})(t) = \bfi_0(t) \bQi(t) \ T(L_{\bfi_0})(t) \, ,
\quad T(L_{\bfi_0})(t_p) = 1 \, ,
\]
where $t_p$ lies in the past of $\bfi_0$. For any time
$t_f$, lying in its future, one then has $T(L_{\bfi_0}) = T(L_{\bfi_0})(t_f)$.
Making use of the Heisenberg commutation relation, the
equation can be solved by standard computations, giving  
\[
T(L_{\bfi_0}) = W(\bfi_0) \ e^{\, -(i/2) \langle \bfi_0, \Delta_D  \bfi_0 \rangle }
\, , \quad \bfi_0 \in \Cc_0 \, .
\]
Here \ $W(\bfi_0) = e^{\, i \int \! dt \, \bfi_0(t) \bQi(t)}$ 
and $\Delta_D$ is the propagator of the differential operator
$K$, defined in Sec.~2. Putting 
\ $\bxi \mapsto F_{\bfi_0}[\bxi] \doteq 
L_{\bfi_0}[\bxi] + (1/2) \langle \bfi_0, \Delta_D  \bfi_0 \rangle$, 
one obtains \mbox{$T(F_{\bfi_0}) = W(\bfi_0)$}, in accordance with the 
definition \eqref{e4.2} in Sec.~4.

\medskip
Trying to extend the time-ordered exponentials by a Dyson expansion
to arbitrary functionals in $\Fc$ would fail due to domain
problems. We therefore determine these operators indirectly by relying on 
the preceding results. Making use of the canonical commutation
relations, we obtain on the domain $\Dom_S$ the equality 
$W(\bfi_0) \bQi(t) W(\bfi_0)^{-1} = \bQi(t) + (\Delta \bfi_0)(t)$,  
$t \in \RR$,  where $\Delta$ is the commutator function, defined in 
Sec.\ 2. Making use of the Dyson expansion, 
it follows that for any bounded functional $F \in \Fc$ one has
\begin{equation} \label{e5.1} 
W(\bfi_0) T(F) W(\bfi_0)^{-1} = T(F^{\Delta \bfi_0}) \, , \quad
\bfi_0 \in \Cc_0 \, .
\end{equation}
This implies 
that for any two loop functions $\bfi_{0}, \bgi_{0} \in \Cc_0$
and bounded functionals $F, G \in \Fc$ 
\begin{align*}
& T(F) W(\bfi_{0}) T(G) W(\bgi_{0}) = 
T(F) T(G^{\Delta \bfi_{0}})  W(\bfi_{0})  W(\bgi_{0}) \\
& = T(F) T(G^{\Delta \bfi_{0}}) W(\bfi_{0} + \bgi_{0}) 
e^{- (i/2) \langle  \bfi_{0}, \Delta \bgi_{0} \rangle } \, . 
\end{align*}
We define now for any loop $\bfi_0$ and bounded functional 
$F$ the unitary operators 
\begin{equation} \label{e5.2} 
\overline{T}(L_{\bfi_0} + F) \doteq T(F^{- \Delta_A \bfi_0}) \, T(L_{\bfi_0}) \, ,
\quad \bfi_0 \in \Cc_0, \ F \in \Fc \, ,
\end{equation}
where $\Delta_A$ is the advanced propagator defined in Sec.~2. 
This ansatz is suggested by a similar relation obtained 
in the framework of \cite[Sec.\ 4]{BuFr}. 
Since any functional in $\Fc$ can uniquely be decomposed into its
bounded and  unbounded parts, these operators are well-defined. 
In order to see that they have the properties of time-ordered 
exponentials, let $(L_{\bfi_{0}} + F)$ lie in the future of
$(L_{\bgi_{0}} + G)$. Then 
\begin{align*}
& \overline{T}(L_{\bfi_{0}} + F) \, \overline{T}(L_{\bgi_{0}} + G)
= T(F^{- \Delta_A \bfi_{0}}) \, T(L_{\bfi_{0}})  \ 
T(G^{- \Delta_A \bgi_{0}}) \, T(L_{\bgi_{0}}) \\
& = T(F^{- \Delta_A \bfi_{0}}) T(G^{- \Delta_A \bgi_{0} + \Delta \bfi_{0}}) \,
T(L_{\bfi_{0}}) T(L_{\bgi_{0}}) \, ,
\end{align*}
where we made use of the fact that the operators $T(L_{\bfi_0})$
and $W(\bfi_0)$ differ only by a phase factor and of relation \eqref{e5.1}. 
Now the shift of functionals by loop functions does not affect their 
localization properties, so we can apply the preceding results about
the causal factorization of the restriction of $\overline{T}$ to 
the bounded, respectively linear functionals, giving
\begin{align*}
T(F^{- \Delta_A \bfi_{0}}) & T(G^{- \Delta_A \bgi_{0} 
+ \Delta \bfi_{0}}) \,
T(L_{\bfi_{0}}) T(L_{\bgi_{0}}) \\ 
& =  T(F^{- \Delta_A \bfi_{0}} + G^{- \Delta_A \bgi_{0} + 
\Delta \bfi_{0}}) \, T(L_{\bfi_{0}} + L_{\bgi_{0}}) \, .
\end{align*}
Since $\bgi_{0}$ has support in the past of $F$,
this holds also for $t \mapsto (\Delta_A \bgi_{0})(t)$.
Similarly, since $\bfi_{0}$ has support in the future of
$G$, this is also true for $t \mapsto (\Delta_R \bfi_{0})(t)$.
Thus, bearing in mind that $\Delta = \Delta_R - \Delta_A$,
it follows from the Dyson expansion of the time-ordered exponentials
that 
\[
T(F^{- \Delta_A \bfi_{0}} + G^{- \Delta_A \bgi_{0} + 
\Delta \bfi_{0}}) = 
T( (F + G)^{-\Delta_A (\bfi_{0} + \bgi_{0})}) \, .
\]
Since $L_{\bfi_{0}} + L_{\bgi_{0}} = L_{(\bfi_{0} + \bgi_{0})}$, this proves
that for any pair of such time-ordered functionals one obtains 
the causal factorization relation 
\[
\overline{T}(L_{\bfi_{0}} + F) \overline{T}(L_{\bgi_{0}} + G) 
= \overline{T}(L_{\bfi_{0}} + L_{\bgi_{0}} + F + G) \, .
\]
By the same argument as in case of bounded functionals, one can show
then that the time-ordered 
exponentials  $\overline{T} : \Fc \rightarrow 
\BHil{S}$ also satisfy the non-linear causality
relation given in the definition of the dynamical groups in Sec.~3. 

\medskip
It remains to show that these operators also satisfy 
the dynamical relations for the given Lagrangean $\Lag_0$.
Picking any loop $\bxi_0 \in \Cc_0$ and functional 
$F \in \Fc$, it follows from the definition of 
the time-ordered operators and the action of phase 
factors on them that 
$\overline{T}(F^{\Delta_A \bxi_0}  + F_{\bxi_0}) = \overline{T}(F) 
\overline{T}(F_{\bxi_0})$. Here $F_{\bxi_0}$ is the functional
defined in equation \eqref{e4.1} for which one has, as was shown thereafter,
$F_{K\bxi_0} = \delta \Lag_0(\bxi_0)$. 
Since $\overline{T}(F_{K\bxi_0}) = W(K\bxi_0) = 1$
and $\Delta_A K \bxi_0 = \bxi_0$, we have thus arrived at
the dynamical relation
\[
\overline{T}(F^{\bxi_0}  + \delta \Lag_0(\bxi_0)) = \overline{T}(F) \, ,
\quad F \in \Fc \, , \ \bxi_0 \in \Cc_0 \, .
\]

\medskip
We conclude that the unitary operators  
$\overline{T} : \Fc \rightarrow \BHil{S}$, defined in \eqref{e5.2}, 
satisfy all relations, characterizing the generating elements of the 
dynamical group $\Group_{\Lag_0}$. So we obtain a representation
of this group on $\Hil_S$, denoted by $\pi_s$. It is fixed by 
the relations 
\begin{equation} \label{e5.3}
\pi_S(S_{\Lag_0}(F)) \doteq \overline{T}(F) \, , \ \ 
\pi_S(S_{\Lag_0}(F)^{-1}) \doteq \overline{T}(F)^{-1} \, ,
\quad F \in \Fc \, , 
\end{equation}
and extends to the norm dense span of the group elements in $\Alg_{\Lag_0}$
by linearity. Since we have equipped this span with the maximal 
C*-norm, it also follows that the representation $\pi_S$ extends by 
continuity to $\Alg_{\Lag_0}$. 

\medskip
It is noteworthy that the representation $\pi_S$  has 
significant continuity properties. We say that a representation $\pi$ of
$\Alg_{\Lag_0}$ is regular if the functions 
$c \mapsto \pi(S_{\Lag_0}(c F))$, \mbox{$c \in \RR$}, are continuous
in the strong operator topology for all $F \in \Fc$. 
That this is the case for the representation $\pi_S$ follows 
for bounded functionals from the Dyson expansion, and for 
linear functionals from well-known properties of the resulting
Weyl operators. The statement  for arbitrary functionals is then 
a consequence of relation~\eqref{e5.2}. Based on the 
results in this section, the following theorem obtains.
\begin{theorem} \label{t5.1}
Let $\Alg_{\Lag_0}$ be the dynamical algebra fixed by 
the non-interacting Lagrangean $\Lag_0$. 
This algebra is represented in the Schr\"odinger 
representation of the position and momentum operators 
$\bQi, \bPi$ by the pair $(\pi_S, \Hil_S)$, where the action of the 
morphism $\pi_S$ on the generating elements of $\Alg_{\Lag_0}$
is given by relation \eqref{e5.3}. One has
\begin{itemize}
\item[(i)] the representation $(\pi_S, \Hil_S)$ is irreducible and 
regular.
\item[(ii)] For any Lagrangean $\Lag$ of the form \eqref{e2.3},
the corresponding dynamical algebra $\Alg_\Lag$ is represented on 
$\Hil_S$ by $\pi \doteq \pi_S \circ \gamma$, where $\gamma$ is the 
injective homomorphism
defined in equation \eqref{e4.3}, mapping $\Alg_\Lag$ into
$\Alg_{\Lag_0}$. The representation $\pi$ is irreducible and
regular; this applies also to its restrictions
$\pi \upharpoonright \Alg_\Lag(\II)$ for any given time interval~$\II$ 
with open interior. 
 \end{itemize}
\end{theorem} 
\begin{proof}
(i) Since $\pi_s(\Alg_{\Lag_0})$ contains the unitary exponentials of the
position and momentum operators (Weyl operators), the 
representation $\pi_S$ is irreducible.
Its regularity is a consequence of the regularity properties 
of the representing operators $\overline{T}$, as explained above. 
(ii) Let $\Fc(\II)$ be the functionals 
having support in any given interval~$\II$. They determine 
a corresponding subalgebra $\Alg_{\Lag}(\II) \subset \Alg_{\Lag}$.
Identifying the generating unitaries of 
$\Alg_{\Lag}(\II)$ and $\Alg_{\Lag_0}(\II)$, it follows from
relation \eqref{e4.3} that the algebras 
$\pi(\Alg_{\Lag}(\II))$ and $\pi_0(\Alg_{\Lag_0}(\II))$ 
on $\Hil_S$ are unitarily 
equivalent. Since the group $\Group_{\Lag_0}(\II)$ contains the
Weyl group, cf.\ Theorem \ref{t4.1}, 
The algebra $\pi_S(\Alg_{\Lag_0}(\II))$ acts irreducibly on 
$\Hil_S$, so this is also true for 
$\pi(\Alg_{\Lag}(\II))$, proving
the irreducibility of these representations. Since the 
scaling of functionals $c \mapsto c F$, $c \in \RR$, does not
affect their support, it also follows from relation \eqref{e5.3}
that all functions $c \mapsto \pi(S_{\Lag}(c F))$, $F \in \Fc$, are continuous
in the strong operator topology  on $\Hil_s$, completing the proof.
\end{proof}

\medskip
Let us mention that it is not known whether the 
representation $(\pi_s, \Hil_s)$ of $\Alg_{\Lag_0}$ is 
faithful. An affirmative answer would be of interest since it would 
imply, by an application of the Stone-von Neumann theorem,  
that $(\pi_s, \Hil_s)$ is, up to equivalence, the unique regular,
irreducible representation of $\Alg_{\Lag_0}$.

\section{Operations and probabilities}
\setcounter{equation}{0}

Temporary operations, which 
are performed on physical systems, are the primary ingredients 
in our setting. The concept of observable was not used until now. 
So there arises the question of whether one can recover from our 
present point of view the standard statistical interpretation of 
quantum physics in terms of observables in an operationally meaningful 
manner. 

\medskip  
In order to discuss this 
issue, let $\Alg_\Lag$ be a dynamical algebra  in a
representation $(\pi, \Hil)$, having properties
established in the preceding theorem. The normalized vectors 
in $\Hil$ are denoted by $\Omega$ and the corresponding vector states 
on~$\Alg_\Lag$ are given by 
$\omega(\, \cdot \,) = \langle \Omega, \,\pi( \cdot )\Omega \rangle$.  
The operations which can be performed on an 
ensemble, described by a vector state 
$\omega$, correspond to maps 
$\omega \mapsto \omega_S \doteq \omega \, \circ \Ad{S^{-1}}$, where 
$S \in \Group_S$ and $\omega_S$ is fixed by the ray of
$\pi(S) \, \Omega$. So the transition probability between the initial
and final state 
is given by $\omega \cdot \omega_S \doteq |\omega(S)|^2$. 

\medskip 
It is not clear which portion of the state space can be reached by 
the action of $\Group_S$ on a given state. In view of the fact that 
this group describes an abundance of perturbations,
it seems possible that it acts (almost) transitively, 
\ie that its range on any given vector state is norm 
dense in the set of all vector states. 
If this is not the case, one may rely on the superposition principle 
and proceed to linear combinations 
$\overline{S} = \sum c_k S_k$ of operations, which are norm dense
in the unitaries of $\Alg_{\Lag_0}$. It then follows from Kadison's
transitivity theorem \cite{Ka1} that this extended unitary group 
acts transitively on all vector states. 
So transition probabilities between pure states 
can be determined by operations, \ie without having to rely 
on the existence of minimal projections.

\medskip
We will show next that there exist operations which, when acting on 
given ensembles, produce vector states with prescribed properties, which  
are described by a projection. We will restrict our attention here to 
projections of infinite dimension since 
this covers the important case of observables having continuous spectrum.
The respective probabilities that members of the original ensembles 
have the given properties is likewise encoded in transition 
amplitudes, as defined in the preceding step. To be precise, 
this  holds true only  up to some given, arbitrarily small error.

\medskip 
Since the representation
$(\pi, \Hil)$ is irreducible, \ie $\pi(\Alg_\Lag)'' = \BHil{{ }}$, 
we can extend the vector states $\omega$  on $\Alg_\Lag$
to its weak closure $\Alg_\Lag^-$
with regard to the  weak operator topology, determined 
by the representation. We then have the following result, 
where we make us of arguments in \cite{BuSt}.

\begin{theorem}
Let $\Hil_N \subset \Hil$ be any finite dimensional subspace, 
let $E \in \Alg_\Lag^-$ be any infinite dimensional projection, 
and let $\varepsilon > 0$.
There exists a unitary operator \mbox{$S_{\varepsilon} \in \Alg_\Lag$} 
such that for any vector state $\omega$, given by a 
vector $\Omega \in \Hil_N$, its image $\omega_{S_\varepsilon}$
under the operation $S_{\varepsilon}$ satisfies 
\[
\omega_{S_\varepsilon}(1 - E) < \varepsilon  
\quad \text{and} \quad | \omega \cdot \omega_{S_\varepsilon} - \omega(E)^2 |
< \varepsilon \, .
\]
\end{theorem}
\begin{proof}
In a first step we show that any isometry $V \in \BHil{{ }}$
with range projection $E$ can be approximated by a series of
unitary operators in the strong operator topology. Let 
$E_k \in  \BHil{{ }}$, $k \in \NN$, 
be an increasing sequence of finite dimensional
projections which converges to $1$ in the strong operator
topology. Putting $V_k = V E_k$, the projections
$(1 - V_k V_k^*)$ and $(1- V_k^* V_k) = (1 - E_k)$ have infinite dimension.
So there exist partial isometries $W_k \in \BHil{{ }}$ such that 
$W_k W_k^* = (1 - V_k V_k^*)$ and $W_k^* W_k = (1 - E_k)$.
It follows that $(V_k + W_k)^* (V_k + W_k) = (V_k + W_k) (V_k + W_k)^* = 1$,
which shows that the sums 
$U_k \doteq (V_k + W_k)$ are unitary operators, $k \in \NN$. It is 
also apparent that $U_k$ converges to $V$ in the strong operator 
topology in the limit of large $k$. 

\medskip
According to the preceding step, the isometries 
$V'_k \doteq V U_k^*$, $k \in \NN$,  have the common range projection $E$.
Since $U_k$ converges strongly to $V$ in the limit of large $k$,
these isometries  
converge in the weak operator topology to $V V^* = E$.
So in particular 
$\lim_k \langle \Omega, V'_k \Omega \rangle = \langle \Omega, E \Omega 
\rangle$. Let us mention 
as an aside that this is the largest \mbox{positive} expectation value 
which can be reached by isometries $V'$ with range projection~$E$. 
It shows,  since $\Hil_N$ is finite dimensional, that there exists 
some isometry $V''$ such that 
$|\langle \Omega, V'' \Omega \rangle - 
\langle \Omega, E \Omega \rangle | < \varepsilon /4 $
for all normalized vectors $\Omega \in \Hil_N$.
Moreover, according to the first step, there exists a  unitary 
operator $U$ for which one has 
$\| (U - V'') \, \Omega \| < \varepsilon / 4$.
Combining these estimates one obtains the bounds
$\big| | \langle \Omega, U \Omega \rangle|^2 - 
\langle \Omega, E \Omega \rangle^2 \big| < \varepsilon$
and $\langle U \Omega, (1-E) U \Omega \rangle < \varepsilon$. 
Since the unitary operator~$U$ in these relations acts on 
vectors in the finite dimensional space $\Hil_N$, it can be replaced 
according to Kadison's transitivity theorem 
by some operator $\pi(S_\varepsilon)$,
where $S_\varepsilon \in \Alg_L$ is unitary. 
The statement then follows from the definition of the
perturbed states $\omega_{S_\varepsilon}$. 
\end{proof}

This theorem shows that for any given property, described by an
infinite projection $E$, and any finite dimensional set of 
vector states $\omega$ 
there exists some unitary operator $S$, interpreted as an  
operation, which has two fundamental 
properties: first, the probability that a state $\omega$
has the property $E$ can be determined from the square root of the
transition probability $\omega \cdot \omega_S$ between the 
state before and after the operation. Second, 
the states $\omega_S$ after the operation have 
property $E$ with arbitrary precision. This holds true without having to rely
on a subjective process of state reduction. So $S$ exactly 
describes what one would expect from a well-designed experiment,
measuring $E$. For this reason unitary operations $S$ with these properties
were called \textit{primitive observables} in \cite{BuSt},  
where the term ``primitive'' implies that they are basic.
In that reference it is also discussed how observables 
composed of orthogonal projections can be determined in a similar manner. 

\medskip
Having seen that the conventional interpretation of quantum mechanics
can be recovered by relying on the concept of operations, let us 
also comment on the second ingredient in our approach: time.
Since, from a macroscopic point of view, the arrow of time is 
unquestionable, it also enters in microphysics since realistic operations
can only be performed one after the other. It is impossible to 
make up for missed operations in the past. This motivated us to 
take the time ordering of operations as a fundamental ingredient
in our approach. In spite of the fact that the number of
available operations decreases in the course of time, this does 
not mean that the information which one can gather by using them
also diminishes. This may be seen from Theorem \ref{t5.1} 
according to which the operations localized in any time interval~$\II$ 
are irreducible in the representations of interest. 
Thus, the repetition of experiments, determining a particular 
property $E$, say, can be described in our setting by operations at any 
instant of time. So our approach provides a fully consistent 
description of quantum mechanics. 

\medskip
Let us mention in conclusion that in relativistic quantum physics,
described by quantum field theory, operations are 
primary ingredients of the theory, as well \cite{BuFr}. Yet instead of 
considering operations in ordered time slices, one has to consider 
there operations in future directed light cones, which are 
partially ordered. It was shown in \cite{BuRo}, that this point of view
leads to a consistent interpretation of the theory. The fact 
that the algebras generated by operations in lightcones are in 
general not irreducible was discussed in \cite{BuSt} and led to the 
concept of primitive observables, used also here.

\section{Summary}
\setcounter{equation}{0}

In this article we have presented an approach to quantum 
mechanics which is entirely based on concepts and facts taken from the 
``classical world''. We proceded from classical mechanics, thinking
of the configuration space of a finite number of particles and of their
motions (orbits). These motions are governed by a given Lagrangean and
the corresponding action. The particles can be perturbed by 
forces, described by functionals involving quite arbitrary 
potentials and some information as to when and
for how long these perturbations act. 

\medskip
We then went on and represented this structure by some dynamical
group, aiming to desribe the effects of perturbations on 
the underlying system. Its generating elements are labelled by 
the functionals, describing the perturbations. Their inverses 
represent the idea that in finite systems it is possible to 
remove the effects of a perturbation by other suitable perturbations.
The dynamics entered into the group by saying how a variation of 
the classical action affects the perturbations. It resulted in a first 
``dynamical'' relation, encoding information about the 
evolution of the system. In a second ``causal" relation, describing the 
ordering effects of time, we made use of the fact that any functional
comprises information as to when the corresponding perturbation takes place. 
This allowed us to incorporate the arrow of time into the group
by relying on the temporal order of perturbations. The group elements  
corresponding to the total effect of two successive 
perturbations, described by the sum of the underlying
functionals, are equal to the product of the 
group elements corresponding to the individual perturbations.
These two basic ingredients, together with a choice of Lagrangean, determine
the structure of the dynamical group. The remaining construction of a 
dynamical C*-algebra then follows from familiar mathematical arguments.

\medskip
It is a remarkable fact that our ``classical approach'', where no 
quantization rules were incorporated from the outset, reproduces 
the structure of quantum mechanics in every respect. As has become
clear by our analysis, the intrinsic non-commutativity of the dynamical 
algebra is a consequence of the arrow of time, which is 
incorporated in our setting. So one could argue that it
is this arrow which is at the origin of the ``quantization'' of  
the classical theory. The specific form of commutation relations then 
follows from the underlying classical dynamics. 

\medskip
We refrain from entering here into these interesting foundational questions.   
But let us mention that our novel approach may be useful also from
a pragmatic point of view. As already mentioned, it was 
discovered in \cite{BuFr} in the framework of 
quantum field theory in an attempt to complement the 
axiomatic framework with some dynamical input; there the  
construction of a dynamical C*-algebra for an interacting Bose field 
was accomplished. 
But that scheme may be applied to the ``quantization'' of 
quite arbitrary classical theories. What is needed is a classical 
configuration space which is invariant under the action of some
group (the loop functions in the present setting), a Lagrangean, 
and some causal order (fixed by time, lightcones in spacetime,  
\etc). One can then go ahead and construct a 
corresponding dynamical C*-algebra in analogy to the examples  
discussed in \cite{BuFr} and the present article. To determine
from it the structure of the resulting quantum theory is then
a matter of computation.


\vspace*{-1mm}
\section*{Acknowledgement}

\vspace*{-2mm}
DB gratefully acknowledges the hospitality extended to him by
Dorothea Bahns and the Mathematics Institute of the University of 
G\"ottingen.

\end{document}